\documentclass[conference]{IEEEtran}
\IEEEoverridecommandlockouts
\usepackage{graphicx}
\usepackage{balance}  
\usepackage{listings}
\usepackage{algorithm}
\usepackage{algorithmicx}
\usepackage{algpseudocode}
\usepackage{capt-of}
\usepackage{amsfonts}
\usepackage{amsmath}
\usepackage{amsthm}

\newcommand{\pluseq}{\mathrel{+}=}

\newtheorem{definition}{Definition}

\IEEEoverridecommandlockouts
\IEEEpubid{\makebox[\columnwidth]{978-1-5386-5541-2/18/\$31.00~\copyright2018 IEEE \hfill} \hspace{\columnsep}\makebox[\columnwidth]{ }}

\begin{document}

\title{Streaming Temporal Graphs: Subgraph Matching\\
\thanks{Sandia National Laboratories University Programs funded this work.}
}

\author{\IEEEauthorblockN{Eric L. Goodman}
\IEEEauthorblockA{\textit{Mission Algorithms Research and Solutions} \\
\textit{Sandia National Laboratories / CU Boulder}\\
Albuquerque, NM, USA \\
elgoodm@sandia.gov}
\and
\IEEEauthorblockN{Dirk Grunwald}
\IEEEauthorblockA{\textit{Department of Computer Science} \\
\textit{CU Boulder}\\
Boulder, CO, USA \\
dirk.grunwald@colorado.edu}
}

\maketitle

\IEEEpubidadjcol



\begin{abstract}
We investigate solutions to subgraph matching within a temporal
stream of data.  We present a high-level language for describing
temporal subgraphs of interest, the \emph{Streaming Analytics Language} (SAL).
SAL programs are translated into C++ code that is run in parallel
on a cluster.  We call this implementation of SAL 
the \emph{Streaming Analytics Machine} (SAM).  SAL programs are succinct, 
requiring about 20 times fewer lines of code than
using the SAM library directly, or writing an implementation using Apache Flink.
To benchmark SAM we calculate finding temporal triangles within streaming
netflow data.  Also, we compare SAM to an implementation written for
Flink.  We find that SAM is able to scale to 128 nodes or 2560 cores, 
while Apache Flink has max throughput with 32 nodes and degrades thereafter.  
Apache Flink has an advantage when triangles are rare, with 
max aggregate throughput for Flink at 32 nodes greater than the max
achievable rate of SAM.  In our experiments, 
when triangle occurrence was faster
than five per second per node, SAM performed better.  Both
frameworks may miss results due to latencies in network communication.
SAM consistently reported an average of 93.7\% of expected 
results while Flink decreases from 83.7\% to 52.1\% as we increase to
the maximum size of the cluster.  Overall, SAM can obtain
rates of 91.8 billion netflows per day.
\end{abstract}

\section{Introduction}
Subgraph isomorphism, the decision problem of determining for a pair of graphs, $H$ and $G$, if $G$ contains a subgraph that
is isomorphic to $H$, is known to be NP-complete.  Additionally, the problem of finding all matching 
subgraphs within a larger graph is 
NP-hard \cite{Lee:2012:ICS:2448936.2448946}.  However, when we 
consider a streaming environment with bounds on the 
temporal extent of the subgraph, the problem
becomes polynomial.  We introduce some definitions to lay 
the groundwork for discussing complexity.

\vspace{-2mm}
\begin{definition}{Graph:} $G = (V, E)$ is a graph where $V$ are vertices 
and $E$ are edges.  Edges $e \in E$ are tuples of the form $(u, v)$ where $u, v \in V$.
\end{definition}

\vspace{-2mm}
\begin{definition}{Temporal Graph:}
$G = (V, E)$ is a \emph{temporal graph} where $V$ are vertices 
and $E$ are temporal edges.  
Temporal edges $e \in E$ are tuples of the form $(u, v, t, \delta)$
where $u,v \in V$ and $t, \delta \in \mathbb{R}$ and represent the 
start time of the edge and its duration, respectively.
\end{definition}

\vspace{-2mm}
\begin{definition}{Streaming Temporal Graph:}
A graph $G$ is a \emph{streaming temporal graph} when $G$ is a temporal graph
where $V$ is finite and $E$ is infinite.
\end{definition}

\vspace{-2mm}
To express subgraph queries against a temporal graph, we need the notion of
temporal contstaints on edges.  First we define some notation.
For an edge $e$, $starttime(e)$ returns the start time of the edge
while $endtime(e)$ returns when the edge ended, i.e. $t + \delta$.
We use the following simple grammar to discuss temporal constraints.
It is sufficient to understand the underlying concepts and examples, 
but could be expanded to enable greater expressibility.

\vspace{-3mm}
{
\scriptsize
\begin{align*}
<Constraint>  & := <EdgeExpr> <Comp> <EdgeExpr> | \\ 
                       &  <ArithmExpr> <Comp> <RealNumber> \\
<EdgeExpr> & := starttime(<EdgeId>) | endtime(<EdgeId>) \\
<ArithmExpr> & := <EdgeExpr> <Op> <EdgeExpr>\\ 
<Op> & := + | - \\
<Comp>  & := < | > | <= | >=
\end{align*}
}%

Now we define \emph{Temporal Subgraph Queries} using temporal
constraints.  

\begin{definition}{Temporal Subgraph Query:}
For a temporal graph $G$, a temporal subgraph query is composed of 
a graph $H$ with edges of the form $(u,v)$, where each of the
edges may have temporal constraints defined.  The result of the query
is all subgraphs $H' \in G$ that are isomorphic to $H$ and that fulfill the
defined temporal constraints.
\end{definition}

Below is an example using SPARQL-like \cite{standard:w3c:sparql} syntax
to express a triangular temporal subgraph query:
\begin{lstlisting}[caption={Temporal Triangle Query},label={listing:triangle_query}]
{
  x1 e1 x2;
  x2 e2 x3;
  x3 e3 x1;
  starttime(e3) - starttime(e1) <= 10;
  startime(e1) < starttime(e2);
  starttime(e2) < starttime(e3);
}
\end{lstlisting}

The first part of the query with the \emph{x} variables defines the 
subgraph structure $H$.  In this example, we have defined a triangle.
The second part of the query defines the temporal constraints.  
Edge \emph{e1} starts before
\emph{e2}, \emph{e2} starts before \emph{e3}, 
and the entire subgraph must occur within a 10 second window. 
For temporal queries where there is a maximum window specified, and also
a maximum number of edges, $d$, then
the problem of finding all matching subgraphs is polynomial.

\newtheorem{theorem}{Theorem}
\begin{theorem}
\label{polynomial}
For a streaming temporal graph $G$, let $n$ be the maximum number of edges
during a time duration of length $\Delta q$.  A temporal subgraph
query $q$ with $d$ edges and maximum temporal extent $\Delta q$ and computed on
an interval of edges of length $\Delta q$ has 
temporal and spatial complexity of $O(n^d)$.
\end{theorem}

\begin{proof}
For an interval of length $\Delta q$, there are at most $n$ edges.  
Assume each of those $n$ edges satisfy one edge of $q$.  
Then there are at most $n$ potential matching subgraphs in the interval.  
Now assume for each of those $n$ potential matching subgraphs,
each of the $n$ edges satisfies another of the edges in $q$.  
Then there are $n^2$ potential matching subgraphs.  If we continue 
for $d$ iterations, we have a max of $n^d$ matching subgraphs.
Thus to find and store all the matching subgraphs, it would require
$O(n^d)$ operations and $O(n^d)$ space to store the results. 
\end{proof}

The point of Theorem \ref{polynomial} is to show that subgraph 
matching on streaming data is tractable,
and that over any window of time that is about the length of the 
temporal subgraph query, computing over
that window has polynomial complexity, as long as the 
number of edges in the query have some maximum value.
Our motivation in exploring subgraph matching on streaming data is that 
it fits the category of finding malicious behavior within network traffic.  
In our experience, cyber subgraphs of interest are generally 
small in size, i.e. $d < 7$, and so the problem can be safely considered
polynomial.

Network traffic can be thought of as streaming temporal graph where the 
vertices are IP addresses and the edges are communications between endpoints.  
For example, a watering hole attack can be thought of as subgraph pattern.
In a watering hole attack, an attacker infects websites 
that are frequented by the targeted community.
When someone from the targeted community accesses the website, their
machine becomes infected and begins communicating with a machine 
controlled by the attacker. 
Figure \ref{figure:wateringhole} illustrates the attack.

\begin{figure}
\centering
\includegraphics[width=.7\linewidth]{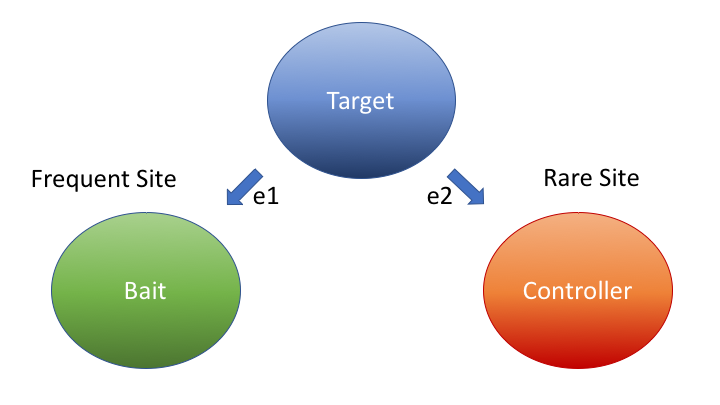}
\caption{Illustration showing the subgraph nature of a watering hole attack.  
The target machine accesses a popular site, \emph{Bait}.  
Shortly after accessing \emph{Bait}, \emph{Target} begins communicating 
to a new machine, the \emph{Controller}}
\label{figure:wateringhole}
\end{figure}
 
Listing \ref{listing:watering_hole} presents the watering hole query. 
This query has both temporal constraints on the edges and
further constraints defined on the vertices, 
namely that \emph{bait} is in \emph{Top1000}, i.e. a popular site, 
and \emph{controller} is not in \emph{Top1000}, i.e. 
a rare or never before seen site.  We use polylogarithmic streaming operators 
to create features for each node.  However, these node-based computations
are outside the scope of this paper, and presented in 
previous work \cite{sal-2019}.  

{
\scriptsize   
\begin{lstlisting}[caption={Watering Hole Query},label={listing:watering_hole}]
{ target e1 bait;
  target e2 controller;
  starttime(e2) > endtime(e1);
  starttime(e2) - endtime(e1) < 20;
  bait in Top1000;
  controller not in Top1000; }
\end{lstlisting}
}

We benchmark two approaches for finding 
temporal subgraphs within a streaming graph.  
One approach we built from the ground up.
We have a domain specific language (DSL) that we call the 
Streaming Analytics Language (SAL), where 
one can express temporal subgraph queries in a 
SPARQL-like\cite{standard:w3c:sparql} 
syntax with temporal constraints on the edges.  
We convert SAL into C++ using the
Scala Parser Combinator \cite{scala-parser-combinator}, which  
allows you to express a grammar of one 
language which is then translated into another language.
The C++ code uses a library we developed, namely the 
Streaming Analytics Machine (SAM).  SAM is a parallel library for 
operating on streaming data that can run in a distributed environment.

The other approach we used builds on Apache Flink \cite{apache-flink}, which 
is a distributed framework purpose built for streaming applications. 
Flink fits many of the needs we have for 
streaming computations and is a 
good basis for comparison.  While we present SAM as an implementation for SAL, 
nothing prevents other implementations of SAL to be developed.

To compare SAM and Flink, we focused on finding temporal triangles, i.e. the subgraph
query expressed in Listing \ref{listing:triangle_query}.  
Temporal triangles was chosen
because they are complex enough to test the frameworks' ability to 
perform difficult computations, and common enough that it could 
occur as a query or part of a query.

Our contributions in this work are the following:
\begin{itemize}
\item We present SAL as a way to quickly express
queries on streaming data.  In this work we focus on the subgraph matching
portion of the language, which is expressed 
using a SPARQL-like \cite{standard:w3c:sparql} syntax.
For the temporal triangle query of Listing \ref{listing:triangle_query}, 
the full SAL program requires 13 lines of code while an implementation in C++
using our custom library takes 238 lines of code and an Apache Flink 
implementation takes 228, representing a significant savings in
writing code.
\item We present an implementation of SAL, using a custom built library 
called the Streaming Analytics Machine, or SAM.  
We show scaling on the temporal triangle
problem to 128 nodes or 2560 cores, 
the maximum size of cluster we could reserve.
The maximum achieved rate is 91.8 billion netflows per day.
\item We compare SAM and Flink, each showing strength in a particular
regime.  SAM outperforms Flink when triangles are frequent.  
\end{itemize}

In the next section we discuss SAL as it pertains to subgraph matching.  In Section
\ref{section:sam}, we discus SAM, the implementation of SAL. 
We then discuss the other approach using Apache Flink in Section \ref{section:flink},
  which is then followed by Section \ref{section:results}, 
which compares SAM and Flink.  Section \ref{section:results2} takes a look at SAM
performance when only some subgraphs match other constraints of the query besides
just the subgraph portion.  Section \ref{section:related} examines related work.Section \ref{section:conclusions} concludes.

\section{Streaming Analytics Language}
\label{section:sal}
The Streaming Analytics Language is a combination of 
imperative statements used for feature extraction
and declarative, SPARQL-like \cite{standard:w3c:sparql} statements, used to
express subgraph queries.  Listing \ref{lst:example}
gives an example program, which is a complete
SAL query expressing a Watering Hole attack.

{

\lstset
{
  basicstyle=\scriptsize,
  xleftmargin=0.03\textwidth,
	numbers=left
}

\begin{lstlisting}[caption=SAL code example,label={lst:example},
                          escapechar=|]
//Preamble Statements
WindowSize = 1000 

// Connection Statements
Netflows = VastStream("localhost", 9999);

// Partition Statements
PARTITION Netflows By SourceIp, DestIp;
HASH SourceIp WITH IpHashFunction;
HASH DestIp WITH IpHashFunction; 

// Defining Features
Top1000 = FOREACH Netflows 
            GENERATE topk(DestIp, 100000, 
                         10000, 1000);

// Subgraph definition
Subgraph on Netflows with source(SourceIp) 
  and target(DestIp)
{
  target e1 bait;
  target e2 controller;
  starttime(e2) > endtime(e1); |\label{lst:example:tc1}|
  starttime(e2) - endtime(e1) < 20; |\label{lst:example:tc2}|
  bait in Top1000;                  |\label{lst:example:vc1}|
  controller not in Top1000;  |\label{lst:example:vc2}|
}	
\end{lstlisting}
}

In this SAL program there are five parts:
1) preamble statements,
2) partition statements,
3) connection statements,
4) feature definition, and
5) subgraph definition.
Preamble statements allow for global constants to be defined that 
are used throughout the program.
In the above listing, line 2 defines the default window size, i.e. the number of items in the sliding window.

After the preamble are the connection statements.
Line 5 defines a stream of netflows called \emph{Netflows}.  
\emph{VastStream} tells the SAL interpreter to expect netflow data of a 
particular format (we use the same format for netflows as found in the 
VAST Challenge 2013: Mini-Challenge 3 dataset \cite{VAST}).  Each
participating node in the cluster receives netflows over a socket on port 9999.
The \emph{VastStream} function creates a stream of tuples that 
represent netflows.  For the tuples generated by \emph{VastStream}, 
keywords are defined to access the individual fields of the tuple.

There are several different standard netflow formats.  SAL 
currently supports only one format, the one used by VAST (\cite{VAST}).  Adding other
netflow formats is a straightforward task.  In fact, SAL can easily be 
extended to process any type of tuple.  Using our current SAL interpreter, 
SAM, the process is to define a C++
std::tuple with the required fields and a function object that accepts 
a string and returns an std::tuple.  Once a mapping is defined from the 
desired keyword (e.g. \emph{VastStream}) to the std::tuple, this new 
tuple type can then be used in SAL connection statements.  The mapping
is defined via the Scala Parser Combinator \cite{scala-parser-combinator} 
and is discussed in more detail in Section \ref{section:sam}.

Following the connection statements is the definition of 
how the tuples are partitioned across the cluster.  Line 8 
specifies that the netflows should be partitioned separately 
by SourceIp and DestIp.  Each node in the cluster
acts as an independent sensor and receives a separate stream of netflows.
These independent streams are then re-partitioned across the cluster.  
In this example, each node is assigned
a set of Source IP addresses and Destination IP addresses.  
How the IP addresses are assigned is using a common hash function.
The hash function used is specified on lines 9 and 10, 
called the \emph{IpHashFunction}. This is another avenue for 
extending SAL.  Other hash functions can be defined
and mapped to SAL constructs, similar to how other tuple
types can be added to SAL.
The process is to define a function object that accepts the tuple type and 
returns an integer, and then map the function object to a keyword 
using the Scala Parser Combinator.

The next part of the SAL program defines features.  On line 13 we use the 
\emph{FOREACH GENERATE} statement to calculate the most 
frequent destination IPs across the stream of netflows.  The
second parameter defines the total window size of the sliding window.
The third parameter defines the size of a basic window 
\cite{topk-sliding-window-2003}, which divides up the sliding
window into smaller chunks.  The fourth parameter defines the number
of most frequent items to keep track of.  

Finally, lines 21 and 22 define the subgraph of interest
in terms of declarative, SPARQL-like statements \cite{standard:w3c:sparql}
that have the form $<$\emph{node}$>$ $<$\emph{edge}$>$ $<$\emph{node}$>$.  
Lines \ref{lst:example:tc1} and \ref{lst:example:tc2} 
define temporal constraints on the edges, specifying that \emph{e1}
comes before \emph{e2} and that from the end of \emph{e1}
to the start of \emph{e2}, the total time is twenty seconds.
Lines \ref{lst:example:vc1} and \ref{lst:example:vc2} define
constraints that the vertices must fulfill, in this case \emph{bait}
must be found in the feature set defined earlier, \emph{Top1000},
while \emph{controller} should not.

In this paper we focus on the subgraph queries, the specification
and implementation which will be discussed in the next section.

\section{Streaming Analytics Machine}
\label{section:sam}

The Streaming Analytics Machine (SAM) is one implementation of
SAL.  We use the Scala Parser Combinator Library \cite{scala-parser-combinator}
to translate SAL code into C++ code which utilizes a parallel library we
developed for expressing distributed streaming programs.

SAM is architected so that each node in the cluster receives tuple data 
Right now for the prototype, the only ingest method is a simple socket layer.  
In maturing SAM, other options  such as Kafka \cite{Garg:2013:AK:2588385} 
is an obvious alternative.  We then use ZeroMQ \cite{Akgul:2013:ZER:2523409} 
to distribute the tuples across the cluster.

For each tuple that a node receives, it performs a hash 
for each key specified in the \emph{PARTITION} statement, and sends the
tuple to the node assigned that key.  For our initial domain of cyber analysis, 
it makes sense to partition netflows by IP address, as many types of analyses 
perform calculations on a per IP basis.  
For each netflow that a node receives, it 
performs a hash of the source IP in the netflow mod the cluster 
size to determine where to send the netflow.  Similarly, for the 
destination IP. 
Thus for each netflow a node receives over the socket layer, 
it sends the same netflow twice
over ZeroMQ.  Each node is then prepared to perform 
calculations on a per IP basis.


For this work, we will focus on the data structures, processes, and algorithms
that enable a subgraph query, such as the one in Listing \ref{listing:triangle_query},
to be expressed in SAM and computed on a stream.
First, the subgraph query is translated 
from the SPARQL-like syntax to
a sequence of C++ objects, namely of type \emph{EdgeExpression}, 
\emph{TimeEdgeExpression}, and \emph{VertexConstraintExpression}.  
For example, the first three lines of Listing \ref{listing:triangle_query}
become EdgeExpressions, in essence, assigning labels
to vertices and edges.  The last three lines of 
Listing \ref{listing:triangle_query} become
TimeEdgeExpressions, utilizing the labels specified in
the EdgeExpressions to define constraints.
VertexConstraintExpression's are constraints on the
vertex itself, and an example can be found in lines
\ref{lst:example:vc1} and \ref{lst:example:vc2} of Listing \ref{lst:example}.
SAM uses the EdgeExpressions and TimeEdgeExpressions
to develop an order of execution based upon the temporal ordering
of the edges.  In the example of Listing \ref{listing:triangle_query},
$e1$ occurs before $e2$, which occurs before $e3$.

The approach we have taken assumes there is a strict temporal
ordering of the edges.  For our intended use case of cyber security, this is generally
an acceptable assumption.  Most queries have causality 
implicit in the question.  For example, the watering hole attack is a sequence
of actions, with the initial cause the malvertisement exploiting a vulnerability
in a target system, which then causes other actions, all of which
are temporally ordered.  As another example, botnets with a command and 
control server, commands are issued from the command and control server,
which then produces subsequent actions by the bots, which again produces
a temporal ordering of the interactions.  Our argument is that requiring
strict temporal ordering of edges is not a huge constraint, or at the very least
covers many relevant questions an analyst might ask of cyber data.
We leave as future work an implementation that can address queries
that do not have a temporal ordering on the edges.
It should be noted that while SAM currently requires temporal ordering of edges,
SAL is not so constrained.  SAM can certainly be expanded to include other
approaches that do not require temporal ordering of edges. 

Once the subgraph query has been finalized with the order of edges 
determined, SAM is now set to find matching subgraphs within
the stream of data.  Before discussing the actual algorithm, we will first
discuss the data structures.

\subsection{Data Structures}

We make use of a number of thread-safe custom data structures, namely
\emph{Compressed Sparse Row} (CSR), 
\emph{Compressed Sparse Column} (CSC), 
\emph{ZeroMQ Push Pull Communicator} (Communicator),
\emph{Subgraph Query Result Map} (ResultMap),
\emph{Edge Request Map} (RequestMap), and
\emph{Graph Store} (GraphStore).

\subsubsection{Compressed Sparse Row}
CSR is common way of representing
sparse matrices. 
The data structure is useful for representing sparse graphs
since graphs can be represented with a matrix, the rows
signifying source vertices and the columns destination vertices.
Instead of a $|V| \times |V|$ matrix, where $|V|$ is the number of vertices,
CSR's have space requirement $O(|E|)$, where $|E|$ is the number of edges.
In the case of sparse graphs, $|E| \ll |V| \times |V|$.

\begin{figure*}
\centering
\begin{minipage}{.30\textwidth}
\centering
\includegraphics[width=.95\linewidth]{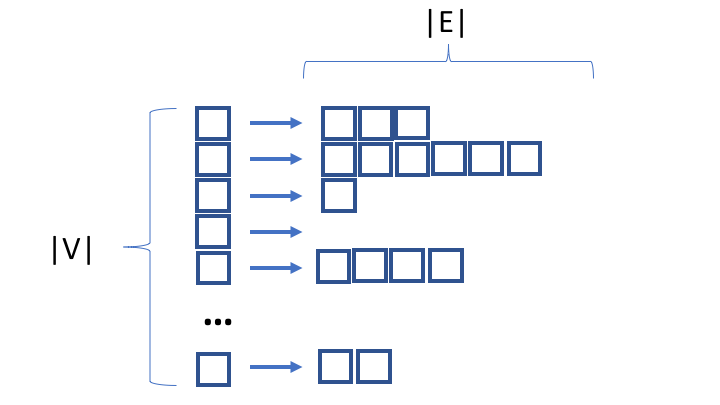}
\captionof{figure}{Traditional Compressed Sparse Row (CSR) data structure.}
\label{figure:csr}
\end{minipage}
\hspace{.01\textwidth}
\begin{minipage}{.30\textwidth}
\centering
\includegraphics[width=.95\linewidth]{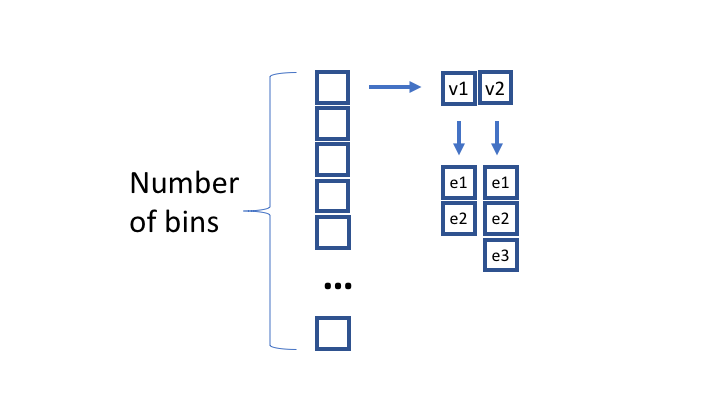}
\captionof{figure}{Modified Compressed Sparse Row (CSR) data structure.}
\label{figure:csr2}
\end{minipage}
\hspace{.01\textwidth}
\centering
\begin{minipage}{.30\textwidth}
\includegraphics[width=.95\linewidth]{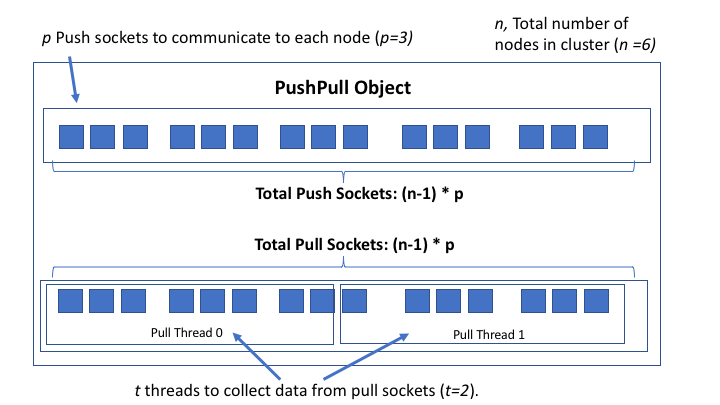}
\caption{The push pull object.}
\label{figure:push_pull_object}
\end{minipage}
\end{figure*}

Figure \ref{figure:csr} presents the traditional CSR data structure.  There
is an array of size $|V|$ where each element of the array points
to a list of edges.  For array element $v$, the list of edges are all the edges
that have $v$ as the source.  With this index structure, we can easily find
all the edges that have a particular vertex $v$ as the source.  However,
a complete scan of the data structure is required to find all edges that have $v$ 
as a destination, leading to the need for the compressed sparse column
data structure, described in the next section.

CSR's as presented work well with static data where the number of vertices
and edges remain constant.  However, our work requires edges that
expire, and also for the possibility that vertices may come and go.
To handle this situation, we create an array of lists of lists of edges, with each
array element protected with a mutex.  Figure \ref{figure:csr2} shows
the overall structure.  Instead of an array of size $|V|$, we create
an array of bins that are accessed via a hash function.  When
an edge is consumed by SAM, it is added to this CSR data structure.
The source of the edge is hashed, and if the source has never been seen,
a list is added to the first level.  Then the edge is added to the list.
If the vertex has been seen before, the edge is added to the existing edge list.
Additionally, SAM keeps track of the longest duration of a registered
query, and any edges that are older than the current time minus the longest
duration are deleted.  Checks for old edges are done anytime
a new edge is added to an existing edge list.

The CSR has mutexes protecting each bin, which
will be important later when
we describe finding subgraph matches.  Regarding performance,
as each bin has its own mutex, the chance of thread contention
is very low.  In metrics gathering, waiting for mutex locks of the CSR
data structure has never been a significant factor in any of our experiments.

\subsubsection{Compressed Sparse Column}
The Compressed Sparse Column (CSC) is exactly the same as the CSR, the only difference
is we index on the target vertex instead of the source vertex.  Both the CSR and CSC are needed,
depending upon the query invoked.  We may need to find an edge with a certain source (CSR), or
an edge with a certain target (CSC).

\subsubsection{ZeroMQ PushPull Communicator}
\label{section:communicator}

\emph{ZeroMQ} handles communication between nodes 
using the push pull paradigm where 
producers push data to push sockets
and consumers pull from corresponding pull sockets.  
The \emph{ZeroMQ PushPull Communicator} (Communicator) 
sends messages to other nodes.  There are push 
sockets for each node in the cluster, except
for the node itself.
There are corresponding pull sockets that pull the data coming from other
nodes.
The Communicator allows for callback functions
to be registered.  These callback functions are invoked whenever a pull 
socket receives data.

ZeroMQ opens multiple ports in the dynamic range (49152
to 65535) for each ZMQ socket that is opened logically in the code.  
Nevertheless, we found
that contention using the ZMQ sockets created significant 
performance bottlenecks.  As such, we created multiple 
push sockets per node, and anytime a node needs to send data to a 
node, it randomly selects one of the push sockets to use.
Figure \ref{figure:push_pull_object} gives an overview
of the design.  If there are $n$ nodes in the cluster and $p$ 
push sockets per node, in total there are
$(n-1)p$ push sockets.

We also found that multiple threads were necessary to collect the data.
$t$ threads are dedicated to pulling data.
In our experiments, we found 4 push sockets per node and
16 total pull threads worked best for most situations.

\subsubsection{Subgraph Query Result Map}
The ResultMap stores intermediate query results with 
There are two central methods, \emph{add} and \emph{process}.  The method
\emph{add} takes as input an intermediate query result, creates additional 
query results based on the local CSR and CSC,
and then creates a list of edge requests for other edges on other nodes 
that are needed to complete queries.  The method \emph{process} is similar, except
in this case it takes an edge that has been received by the node and it checks to see
if that edge satisfies any existing intermediate subgraph queries.  The details
of processing intermediate results and edges used will be discussed in greater 
depth in Section \ref{section:algorithm}.

\begin{figure}
\centering
\includegraphics[width=.99\linewidth]{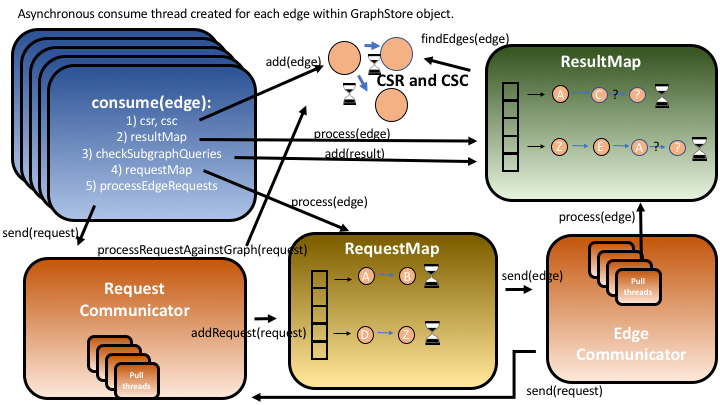}
\caption{An overview of the algorithm employed to detect subgraphs of interest.
Asynchronous threads are launched to consume each edge.  The consume threads
then performs five different operations.  Another set of threads are pulling requests
from other nodes for edges within the RequestCommunicator.  Finally, one more
set of threads pull edges from other nodes in reply to edge requests.}
\label{figure:algorithm}
\end{figure} 

The ResultMap uses a hash structure to store
results.  It is an array of vectors of intermediate query results.  The array is of fixed size, so
must be set appropriately large to deal with the number of results generated during
the max time window specified for all registered queries, otherwise excessive time
is spent linearly probing the vector of intermediate results.  
There are mutexes protecting each access to each array slot.  For each intermediate 
result created, it is indexed by the source, the target, 
or a combination of the source and target, depending on what the query needs
to satisfy the next edge.  The index is the result of a hash function, and that index
is used to store the intermediate result within the array data structure.

\subsubsection{Edge Request Map}
The RequestMap is the object that receives 
requests for edges that match certain criteria, and sends out 
edges to the requestors whenever an
edge is found that matches the criteria.  There are two important methods to mention:
\emph{addRequest(EdgeRequest)} and \emph{process(Edge)}.  The
\emph{addRequest} method takes an \emph{EdgeRequest} object and stores it.
The \emph{process} method takes an \emph{Edge} and checks to see if 
there are any registered edge requests that match the given edge.

The main data structures is again a hash table.  The approach to hashing is similar to
the ResultMap, as it depends on whether the \emph{source}
is defined in the request data structure, the \emph{target} is defined, or both.  
The index is then used to find a bin with an array of lists of EdgeRequest objects.
Similar to the ResultMap, whenever \emph{process} is called,
edge requests that have expired are removed from the list.

\subsubsection{Graph Store}
The GraphStore object uses all the previous data 
structures, orchestrating them together to perform subgraph queries.  
It has a pointer to the ResultMap
and to the RequestMap.  It also has pointers to two Communicators, 
one for sending/receiving edges to/from 
other nodes (Edge Communicator), and another
for sending/receiving edge requests (Request Communicator).  
GraphStore maintains the CSR and CSC
data structures.  How all these pieces work together is described in the next section.   
 
\vspace{10mm} 
\subsection{Algorithm}
\label{section:algorithm}

The algorithm employed to find matching subgraphs employs three sets of threads:
1) GraphStore consume threads, 
2) Request Communicator threads, and
3) Edge Communicator threads.

\subsubsection{GraphStore Consume(Edge) threads}
\label{section:graphstore_threads}
The GraphStore object's primary method is the \emph{consume}
function.  GraphStore is tied into a producer, and any
edges that the producer generates is fed to the \emph{consume}
function.  For each \emph{consume} call, an asynchronous thread is 
launched.  We found that occasionally a call to \emph{consume}
would take an exorbitant amount of time, orders of magnitude greater
than the average time.  These outliers were a result of delays in the
underlying ZeroMQ communication infrastructure, proving difficult to
entirely prevent.  As such, we added the asynchronous approach,
effectively hiding these outliers by allowing work to continue even if
one thread was blocked by ZeroMQ.

Each consume thread performs the following actions:
1) Adds the edge to the CSR and CSC data structures.
2) Processes the edge against the ResultMap,
looking to see if any intermediate results can be further developed
with the new edge.
3) Checks to see if this edge satisfies the first edge of any registered
queries.  If so, adds the results to the ResultMap.
4) Processes the edge against the RequestMap to
see if there are any outstanding requests that match with the new edge.
5) Steps two and three can result
in new edges being needed that may reside on other nodes.  These
edge requests are accumulated and sent out.

For steps 2-4, we provide greater detail below:

\textbf{Step 2: ResultMap.process(edge)}:
ResultMap's \emph{process(edge)} function
is outlined in Algorithm \ref{algorithm:result_map_process}.  At the top level,
the \emph{process(edge)} dives down into three
\emph{process(edge, indexer, checker)} calls.
Intermediate query results are indexed by the next edge to be matched.
The \emph{srcIndexFunc}, \emph{trgIndexFunc}, \emph{bothCheckFunc}
are lambda functions that perform a hash function against the appropriate
field of the edge to find where relevant intermediate results are stored
within the hash table of the ResultMap.
The $srcCheckFunc$ is a lambda function that takes as input a
query result.  It checks that the query result's next edge has a bound source
variable and an unbound target variable.  Similarly, $trgCheckFunc$ returns
true if the query result's next edge has an unbound source variable
and bound target variable.  The $bothCheckFunc$ returns true if both 
variables are bound.  These lambda functions are passed
to the \emph{process(edge, indexer, checker)} function, where the logic
is the same for all three cases by using the lambda functions.

The \emph{process(edge, indexer, checker)} function (line \ref{rmp:process2}, 
uses the provided \emph{indexer} to find the location (line \ref{rmp:indexer}) 
of the potentially relevant intermediate results in the hash 
table (class member \emph{alr}).  Line \ref{rmp:time} finds the time
of the edge.  We use this on line \ref{rmp:delete} to judge whether an 
intermediate result has expired and can be deleted.
On line \ref{rmp:for_all_results} we iterate through all the intermediate
results found in the bin $alr[index]$.  We use the $checker$ lambda 
function to make sure the intermediate result is of the expected
form, and if so, we try to add an edge to the result.  The
$addEdge$ function  (line \ref{rmp:addedge}) adds an edge to the
intermediate result if it can, and if so, returns 
a new intermediate result but
leaves the pre-existing result unchanged.  This allows the 
pre-existing intermediate result to match with other edges later.
On line \ref{rmp:pushback} we add the new result to the list, $gen$,
of newly generated intermediate results.  On line
\ref{rmp:process_against_graph} we pass $gen$
to \emph{processAgainstGraph}.

The \emph{processAgainstGraph} function is used to see if there
are existing edges in the CSR or CSC that can be used
to further complete the queries.  The overall approach
is to generate frontiers of modified results, and continue
processing the frontier until no new intermediate results
are created.  Line \ref{rmp:frontier_start} sets the
\emph{frontier} to be the beginning of \emph{gen}'s 
iterator.  Then we push back a null element onto \emph{gen}
on line \ref{rmp:frontier_null} to mark the end of the current
frontier.  Line \ref{rmp:while1} continues iterating until
there are no longer any new results generated.  Line
\ref{rmp:while2} iterates over the frontier until a null
element is reach.  Lines \ref{rmp:csr} and \ref{rmp:csc}
find edges from the graph that match the provided
intermediate result on the frontier.  Then on line
\ref{rmp:forall_found_edges} we try to add the edges
to the intermediate result, creating new results
that are added to the new frontier on line \ref{rmp:newR}.
Eventually the updated list of intermediate results is returned
to the \emph{process(edge, indexer, checker)} function.
This list is then added to the hash table on line \ref{rmp:add_nocheck}.
The function \emph{add\_nocheck} is a private method that differs
from the public \emph{add} in that it doesn't check the CSR
and CSC, since that step has already been taken.  It also
generates edge requests for edges that will be found on other
nodes, and that list is returned on line \ref{rmp:return_requests2}
and then in the overarching \emph{process} function on line
\ref{rmp:return_requests1}.

{

\begin{algorithm}
\scriptsize
\caption{ResultMap: Processing a New Edge}
\label{algorithm:result_map_process}
\begin{algorithmic}[1]
\Function{ResultMap.process}{$edge$}  \label{rmp:process1}
  \State $requests$, a list of edge requests.
  \State $requests \pluseq process(edge, srcIndexFunc$
  \State \hspace{35 mm}$srcCheckFunc)$ 
  \State $requests \pluseq process(edge, trgIndexFunc,$
  \State \hspace{35 mm}$trgCheckFunc)$
  \State $requests \pluseq process(edge, bothIndexFunc,$
  \State \hspace{35 mm}$bothCheckFunc)$
  \State\Return $requests$ \label{rmp:return_requests1}  
\EndFunction
\State
\Function{ResultMap.process}{$edge$, $indexer$, $checker$} \label{rmp:process2}
  \State $gen$, a list of intermediate results generated.
  \State $index \gets indexer(edge)$ \label{rmp:indexer} 
  \State $currentTime \gets getTime(edge)$ \label{rmp:time}
  \State $mutexes[index].lock()$
  \State $requests$, a list of edge requests.
  \ForAll {$r \in alr[index]$} \label{rmp:for_all_results}
    \If {$r.isExpired(currentTime)$} 
      \State $erase(r)$  \label{rmp:delete}
    \Else
      \If {$checker(r)$}
        \State $newR \gets r.addEdge(edge)$ \label{rmp:addedge}
        \If {$newR$ $not$ $null$}
          \State $gen.push\_back(newR)$ \label{rmp:pushback}
        \EndIf 
      \EndIf
    \EndIf
  \EndFor
  \State $mutexes[index].unlock()$
  \State $gen \gets processAgainstGraph(gen)$ \label{rmp:process_against_graph}
  \ForAll {$g \in gen$}
    \State $requests \pluseq add\_nocheck(g)$  \label{rmp:add_nocheck}
  \EndFor
  \State\Return $requests$ \label{rmp:return_requests2}
\EndFunction
\State
\Function{ResultMap.processAgainstGraph}{$gen$}
\State $frontier \gets gen.begin()$ \label{rmp:frontier_start}
\State $gen.push\_back(null)$ \label{rmp:frontier_null}
\While{$frontier \neq gen.end()$ } \label{rmp:while1}
  \While{$frontier \neq null$}          \label{rmp:while2}
    \If {$!frontier.complete()$}
      \State \texttt{++}$frontier$
      \State $foundEdges = csr.findEdges(frontier)$\label{rmp:csr}
      \State $foundEdges \pluseq csc.findEdges(frontier)$\label{rmp:csc}
      \ForAll {$edge \in foundEdges$}  \label{rmp:forall_found_edges}
        \State $newR \gets frontier.addEdge(edge)$
        \If {$newR$ $not$ $null$}
          \State $gen.push\_back(newR)$ \label{rmp:newR}
        \EndIf
      \EndFor
     \EndIf
  \EndWhile
  \State $frontier \gets frontier.erase()$
  \State $gen.push\_back(null)$
\EndWhile
\State\Return $gen$
\EndFunction
\end{algorithmic}
\end{algorithm}

}

\textbf{Step 3: Checking registered queries}:
When a GraphStore object receives a new edge, it must check to see if
that edge satisfies the first edge of any registered queries.
The process is outlined in Algorithm \ref{algorithm:check_queries}.
The \emph{GraphStore.}
\emph{checkQueries(edge)} function is 
straightforward.  On line \ref{cq:forall_queries} we iterate
through all registered queries, checking to see if the edge
satisfies the first edge of the query (line \ref{cq:satisfies}),
and if so, adds a new intermediate result to the 
ResultMap (line \ref{cq:add}). 

The ResultMap's \emph{add} function begins on line
\ref{cq:rm_add}.  We may create intermediate results
where the next edge belongs to another node, so
we initialize a list, $requests$, to store those requests (line \ref{cq:list_requests}).
We need to check the CSR and CSC if the new result
can be extended further by local knowledge of the graph.
As such we make a call to \emph{processAgainstGraph}
on line \ref{cq:process_against_graph}.
On line \ref{cq:all_results} we iterate through all
generated intermediate results.  If the result is not complete,
we calculate an index (line \ref{cq:hash}), and
place it within the array of lists of results (line \ref{cq:push_back})
with mutexes to protect access.  If the result is complete,
we send it to the output destination (line \ref{cq:output}).
At the very end we return a list of requests, which is then
aggregated by GraphStore's \emph{checkQueries} function,
and then returned on line \ref{cq:return} and eventually used 
by step 5 to send out the requests to other nodes.

\begin{algorithm}
\scriptsize
\caption{Checking Registered Queries}
\label{algorithm:check_queries}
\begin{algorithmic}[1]
\Function{GraphStore.checkQueries}{$edge$}  \label{cq:check_queries}
  \State $requests$, a list of edge requests. 
  \ForAll {$q \in queries$} \label{cq:forall_queries}
    \If {$q.satisfiedBy(edge)$} \label{cq:satisfies}
      \State $requests \pluseq resultMap.add(Result(edge))$\label{cq:add}
    \EndIf
  \EndFor
  \State\Return $requests$ \label{cq:return}
\EndFunction
\State
\Function{RequestMap.add}{$result$} \label{cq:rm_add}
  \State $requests$, a list of edge requests. \label{cq:list_requests}
  \State $gen$, a list of intermediate results generated.
  \State $gen.push\_back(result)$
  \State $gen \gets processAgainstGraph(gen)$\label{cq:process_against_graph}
  \ForAll {$g \in gen$} \label{cq:all_results}
    \If {$!g.complete()$} \label{cq:complete}
      \State $index = g.hash()$ \label{cq:hash}
      \State $mutexes[index].lock()$
      \State $alr[index].push\_back(g)$ \label{cq:push_back}
      \State $mutexes[index].unlock()$
    \Else
      \State $output(g)$ \label{cq:output}
    \EndIf
  \EndFor
  \State\Return $requests$
\EndFunction
\end{algorithmic}
\end{algorithm}

\textbf{Step 4: RequestMap.process(edge)}
The purpose of RequestMap's \emph{process(edge)} is to take a new
edge and find all open edge requests that match that edge.  The process
is detailed in Algorithm \ref{algorithm:request_map_process}.
The opening logic of the \emph{RequestMap.process(edge)} function is similar to
that of ResultMap's \emph{process} function.  There are three
lambda functions for indexing into RequestMap's hash table structure,
which is an array of lists of EdgeRequests ($ale$).   There are also three 
lambda functions for checking that an edge matches a request.  
The indexer and the checker for each of the three cases are each
passed to the \emph{process(edge, indexer, checker)} function
on lines \ref{emp:details1} - \ref{emp:details3}.  
The \emph{process(edge, indexer, checker)} function 
iterates through all of the edge requests with the same index
as the current edge (line \ref{emp:for_all_requests}), 
deleteing all that have expired (line \ref{emp:delete}).
If the edge matches the request, the edge is then sent
to another node using the EdgeCommunicator (line \ref{emp:send}).

{

\begin{algorithm}
\scriptsize
\caption{RequestMap: Processing a New Edge}
\label{algorithm:request_map_process}
\begin{algorithmic}[1]
\Function{RequestMap.process}{$edge$}  \label{emp:process1}
  \State $process(edge, srcIndexFunc, srcCheckFunc)$ \label{emp:details1}
  \State $process(edge, trgIndexFunc,trgCheckFunc)$
  \State $process(edge, bothIndexFunc,bothCheckFunc)$  \label{emp:details3}
\EndFunction
\State
\Function{RequestMap.process}{$edge$, $indexer$, $checker$} \label{emp:process2}
  \State $index \gets indexer(edge)$ \label{emp:indexer} 
  \State $currentTime \gets getTime(edge)$ \label{emp:time}
  \State $mutexes[index].lock()$
  \ForAll {$e \in ale[index]$} \label{emp:for_all_requests}
    \If {$e.isExpired(currentTime)$} 
      \State $erase(e)$  \label{emp:delete}
    \Else
      \If {$checker(e, edge)$}
        \State $edgeCommunicator.send(edge)$ \label{emp:send}
      \EndIf
    \EndIf
  \EndFor
  \State $mutexes[index].unlock()$
\EndFunction
\end{algorithmic}
\end{algorithm}

}

\subsubsection{RequestCommunicator Threads}
Another group of threads are the pull threads of the
RequestCommunicator.  These threads pull from
the ZMQ pull sockets dedicated to gathering
edge requests from other nodes.  One callback
is registered, the \emph{requestCallback}.  The
\emph{requestCallback} calls two methods:
\emph{addRequest} and \emph{processRequestAgainstGraph}.
The \emph{addRequest} method registers the request
with the RequestMap.  The method
\emph{processRequestAgainstGraph} checks to see if
any existing edges are already stored locally that
match the edge request.

\subsubsection{EdgeCommunicator Threads}
The last group of threads to discuss are pull threads
of the EdgeCommunicator and its callback function,
\emph{edgeCallback}.  Upon receiving an edge through
the callback, \emph{ResultMap.process(edge)} is called,
which is discussed in detail in section \ref{section:graphstore_threads}.
The method \emph{process(edge)} produces edge requests,
which are then sent via the RequestCommunicator, the same
as Step 5 of \emph{GraphStore.consume(edge)}.

\section{Apache Flink}
\label{section:flink}
We have currently mapped SAL into SAM as the implementation.  SAL is converted
into SAM code using the Scala Parser Combinator \cite{scala-parser-combinator}.  
However, we wanted to compare how another framework would perform.
It is certainly possible to map SAL into other languages and frameworks.
A prime candidate for evaluation is Apache Flink \cite{apache-flink}, which was 
custom designed and built with streaming applications in mind.  In this section,
we outline the approach taken for implementing the triangle query of Listing 
\ref{listing:triangle_query} using Apache Flink.  Once the groundwork has been 
laid, the actual specification of the algorithm is relatively succinct.  The java code
using Flink 1.6 is presented in Listing \ref{listing:flink}.

{
\lstset
{
  basicstyle=\scriptsize,
  xleftmargin=0.03\textwidth,
  language=java,
  numbers=left
}

\begin{lstlisting}[caption={Apache Flink Triangle Implementation},
                                         label={listing:flink},
                                         escapechar=|]
final StreamExecutionEnvironment env = |\label{flink:env_beg}|
  StreamExecutionEnvironment 
  .getExecutionEnvironment();
env.setParallelism(numSources); |\label{flink:env_num_sources}|
env.setStreamTimeCharacteristic(
  TimeCharacteristic.EventTime);   |\label{flink:env_end}|

NetflowSource netflowSource = |\label{flink:datastream1}| 
  new NetflowSource(numEvents, numIps, rate);
DataStreamSource<Netflow> netflows = 
  env.addSource(netflowSource); |\label{flink:datastream2}|

DataStream<Triad> triads = netflows |\label{flink:triads1}|
  .keyBy(new DestKeySelector())
  .intervalJoin(netflows.keyBy( |\label{flink:interval_join1}|
     new SourceKeySelector()))
  .between(Time.milliseconds(0), 
    Time.milliseconds((long) queryWindow * 1000)) |\label{flink:time_window}|
  .process(new EdgeJoiner(queryWindow));|\label{flink:triads2}|
       
DataStream<Triangle> triangles = triads |\label{flink:triangles1}|
  .keyBy(new TriadKeySelector())
  .intervalJoin(netflows.keyBy(
    new LastEdgeKeySelector()))
  .between(Time.milliseconds(0), 
    Time.milliseconds((long) queryWindow * 1000))
   .process(new TriadJoiner(queryWindow)); |\label{flink:triangles2}|

triangles.writeAsText(outputFile, 
  FileSystem.WriteMode.OVERWRITE);
\end{lstlisting}

}

Lines \ref{flink:env_beg} - \ref{flink:env_end}| define the streaming
environment, such as the number of parallel sources (generally the number
of nodes in the cluster) and that event time of the netflow
should be used as the timestamp (as opposed to ingestion time
or processing time).  Lines \ref{flink:datastream1} - \ref{flink:datastream2}
create a stream of netflows.  \emph{NetflowSource} is a custom
class that generates the netflows in the same manner as was used
for the SAL/SAM experiments.  Lines \ref{flink:triads1} - \ref{flink:triads2}
finds triads, which are sets of two connected edges, that fulfill 
the temporal requirements.  The variable \emph{queryWindow}
on line \ref{flink:time_window} was set to be ten in our 
later experiments, corresponding to a temporal window of 10 seconds.

Line \ref{flink:interval_join1} uses 
an interval join.  Flink has different types of joins.  The
one that was most appropriate for our application is the interval
join.  It takes elements from two streams (in this case the same 
stream \emph{netflows}), and defines an interval of time
over which the join can occur.  In our case, the interval
starts with the timestamp of the first edge
and it ends \emph{queryWindow} seconds later.  
This join is a self-join: it merges the \emph{netflows} data stream
keyed by the destination IP with the \emph{netflows}
data stream keyed by the source IP.  In the end we have a stream
of pairs of edges with three vertices of the form 
$A \rightarrow B \rightarrow C$.

Once we have triads that fulfill the temporal constraints, we can
define the last piece: finding another edge that completes
the triangle, $C \rightarrow A$.  Lines \ref{flink:triangles1} - \ref{flink:triangles2}
performs another interval join, this time between the
triads and the original \emph{netflows} data streams.
The join is defined to occur such that the time difference
between the starting time of the last edge and the starting time
of the first edge must be less than \emph{queryWindow} seconds.

Overall, the complete code is 228 lines long.  This number includes
the code that would likely need to be generated were a mapping
betwen SAL and Flink completed.  This includes classes to do
key selection (\emph{SourceKeySelector}, \emph{DestKeySelector},
\emph{LastEdgeKeySelector}, \emph{TriadKeySelector}),
the \emph{Triad} and \emph{Triangle} classes, and classes to
perform joining (\emph{EdgeJoiner} and \emph{TriadJoiner}).
However, the code for the \emph{Netflow} class and the 
\emph{NetflowSource} were not included, as those classes
would likely be part of the library as they are with SAM. 
For comparison, the SAL version of this query is 13 lines long.

\section{Comparison: SAM vs Flink}
\label{section:results}

Here we compare the performance of SAM with the custom Flink implementation.  
In both cases, we run the triangle query of 
Listing \ref{listing:triangle_query}.
We examine weak scaling, wherein the number of elements per node is 
kept constant as we increase the number of nodes.  In all runs, 
each node is fed 2,500,000 netflows.   
For each netflow we randomly select the source and destination IPs 
from a pool of $|V|$ vertices for each netflow generated.  
We find the highest rate achievable by each approach by increasing 
the rate of edges production until the method to can't keep pace.    

For all of our experiments we use Cloudlab \cite{RicciEide:login14},
a set of clusters distributed across three sites, Utah, Wisconsin, and South Carolina, where 
researchers can provision a set of nodes to their specifications.  
In particular we made use of the Wisconsin cluster using node type \emph{c220g5},
which have two Intel Xeon Silver 4114 10-core CPUs at 2.20 GHz and 192 GB ECC DDR4-2666
memory per node.  We vary the number of nodes to be 16, 32, 64, and 120.  
In Section \ref{section:results2} we present scaling to 128 nodes for SAM,
but for this batch of runs it was difficult to get a reservation for a full 128 nodes.

\begin{figure*}[ht]
\centering
\begin{minipage}{.29\textwidth}
\includegraphics[width=.99\linewidth]{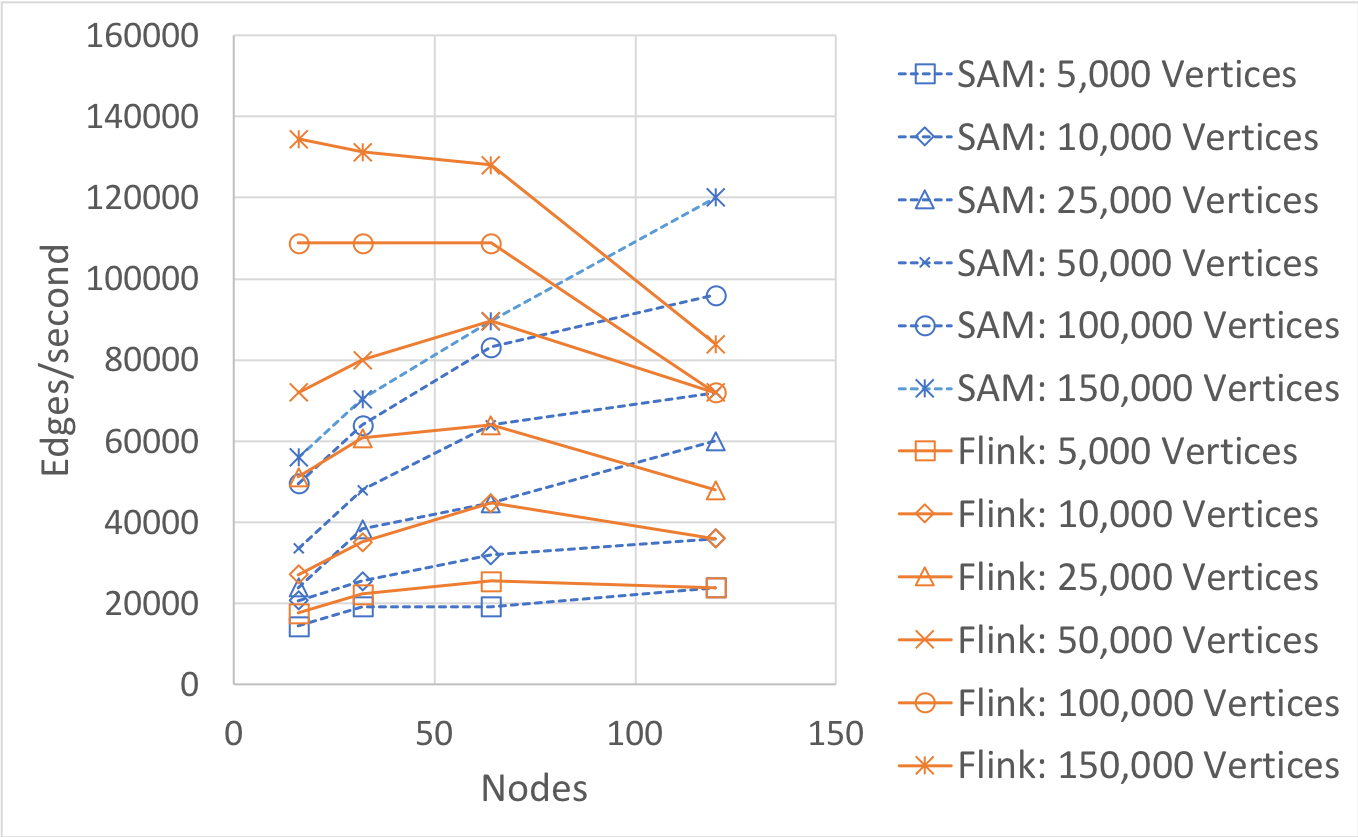}
\captionof{figure}{Shows the maximum throughput for SAM and Flink
where the latency is within 15 seconds of termination. }
\label{figure:max_aggregate_rate_lax}
\end{minipage}
\hspace{.04\textwidth}
\centering
\begin{minipage}{.29\textwidth}
\centering
\includegraphics[width=.99\linewidth]{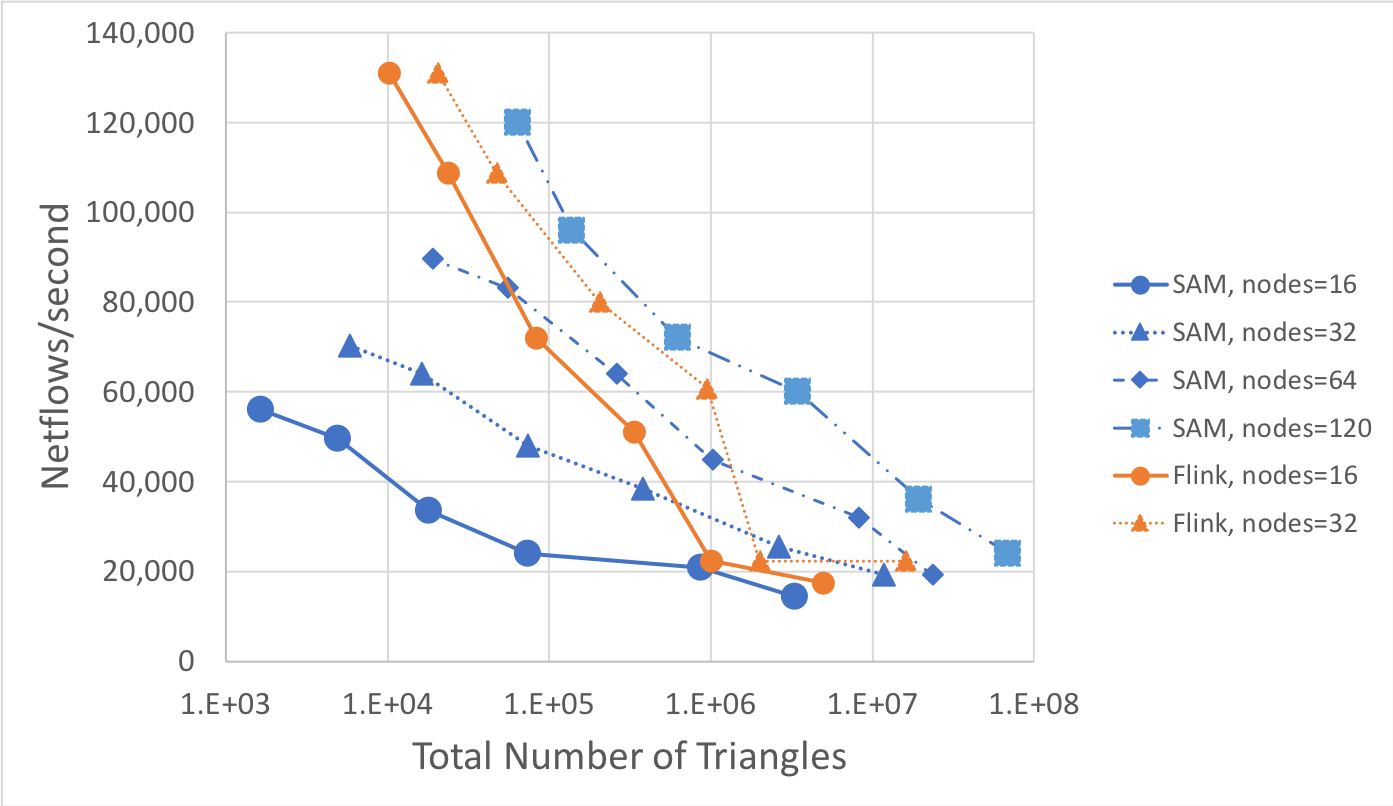}
\captionof{figure}{Presents throughput as a function of total number
of triangles produced.}
\label{figure:netflowVsTriangles}
\end{minipage}
\hspace{.04\textwidth}
\centering
\begin{minipage}{.29\textwidth}
\includegraphics[width=.8\linewidth]{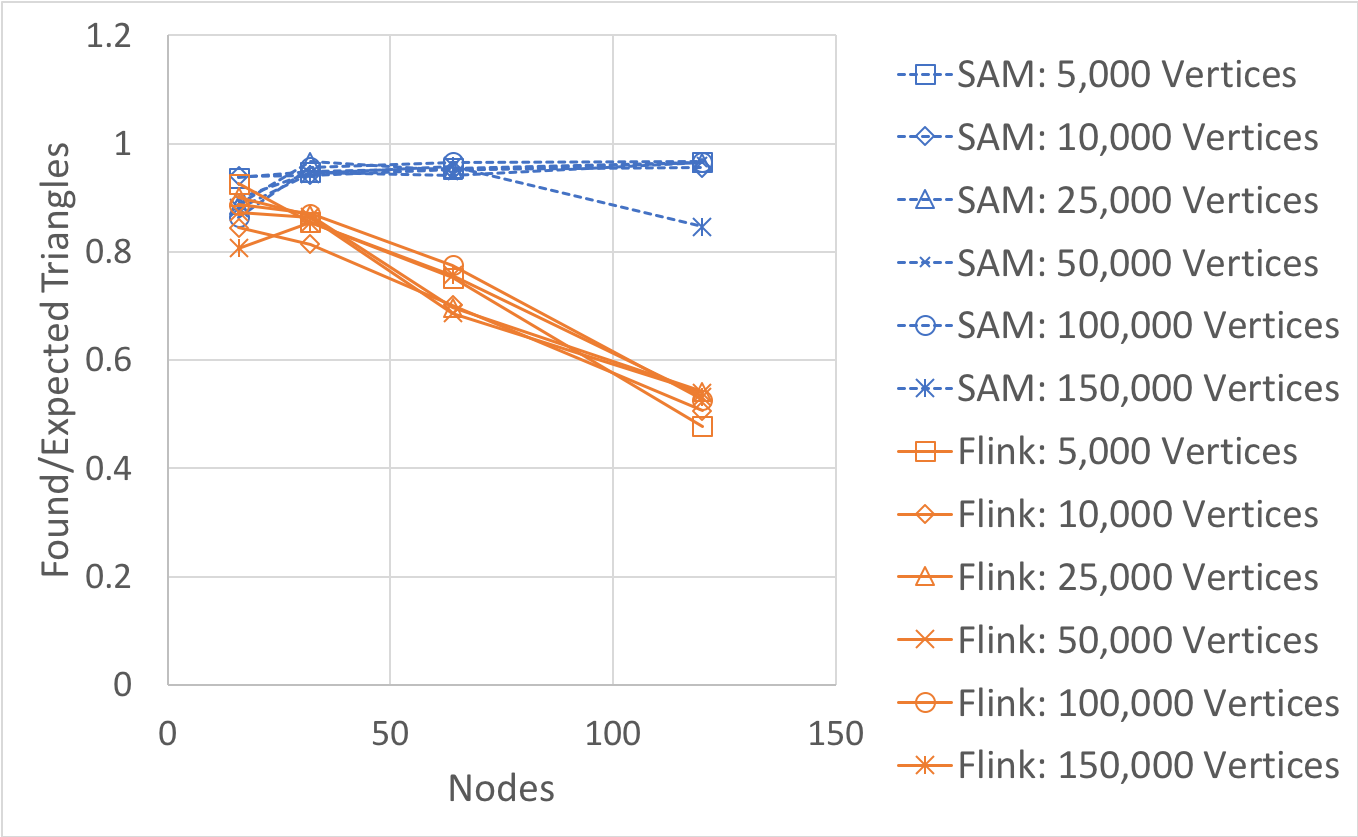}
\caption{This graph shows the ratio of found triangles divided by the number
of expected triangles. 
}
\label{figure:triangle_ratio}
\end{minipage}
\end{figure*}

For the Apache Flink runs we set both the job manager and task manager
heap sizes to be 32 GBs. 
The number of taskmanager task slots was set to 1.  As we only ever ran one Flink job at a time, this seemed to be the most appropriate setting.  The number of sources was the size
of the cluster.

For both implementations,
we kept increasing throughput until run time exceeded $|V| / r + 15$ seconds, 
where $r$ is edges produced per second.  
Latencies longer than 15 seconds indicated that the 
implementation was unable to keep pace.    

Figure \ref{figure:max_aggregate_rate_lax} shows the overall 
max aggregate throughput achieved by SAM and Flink.  We varied the 
number of vertices, $|V|$, to be 5,000, 10,000, 25,000, 
50,000, 100,000, and 150,000.  
Figure \ref{figure:netflowVsTriangles} presents the same data with
another perspective, i.e. the throughput rate as a function of the total
number of triangles produced.
For $|V| = 5,000$ or $10,000$, 
SAM obtained the best overall rate, continuing to garner increasing 
aggregate throughput to 120 nodes. Flink struggled with more than 32 nodes.  
While Flink continued to report results, the percentage of expected triangles 
fell dramatically to about 60-70\% for 64 nodes, and around 50\% for 120 nodes,
as can be seen in Figure \ref{figure:triangle_ratio}.
While Flink was not able to scale past 32 nodes, it showed 
the best overall throughput for $|V| = 50,000$ and $150,000$.  
$|V| = 25,000$ was near the middle ground, with Flink performing 
slightly better, with 32 nodes obtaining an aggregate rate of 60,800
edges per second while SAM obtained 60,000 edges per second on 120 nodes.


To summarize, SAM and Flink have competing strengths.
SAM scales to 120 nodes and excels when the triangle formation 
rate is greater than 5 per second per node.
Flink does better when triangles are rare, but only scales
to 32 nodes.  Both show some loss due to network latencies.
If some data arrives too late, it is not included in the calculations for
either SAM or Flink.  However, SAM has an advantages in this regard,
showing consistent results throughout the processor count range which
is higher than Flink.

\section{Simulating other Constraints}
\label{section:results2}

SAL can be used not only for selecting subgraphs
based on topological structure, but also on vertex  
constraints (e.g. see Listing \ref{lst:example}).
Here 
we simulate vertex constraints by randomly dropping
the first edge of the triangle query, i.e. we process the netflow
through SAM, but we force the edge to not match the first edge
of the triangle query, giving an indication of
SAM's performance with selective queries.  

We compare SAM's performance for when we keep
the first edge all of the time, 1\% of the time, and .1\% 
of the time.  
Figure \ref{figure:keep_queries}
presents the results.  The general trends are as expected.
As we increase $|V|$, the number of triangles created decreases,
and the rate goes up.  Also, as we increase the rate at which edges
are kept causes a decrease in the maximum achievable rate, again
because more triangles are being produced, which causes
a greater computational burden.  The maximum achievable
rate is 1,062,400 netflows per second, when we keep
0.01\% of the first edges and $|V|=150,000$, or 91.8 billion per day.

\begin{figure}
\centering
\centering
\includegraphics[width=.75\linewidth]{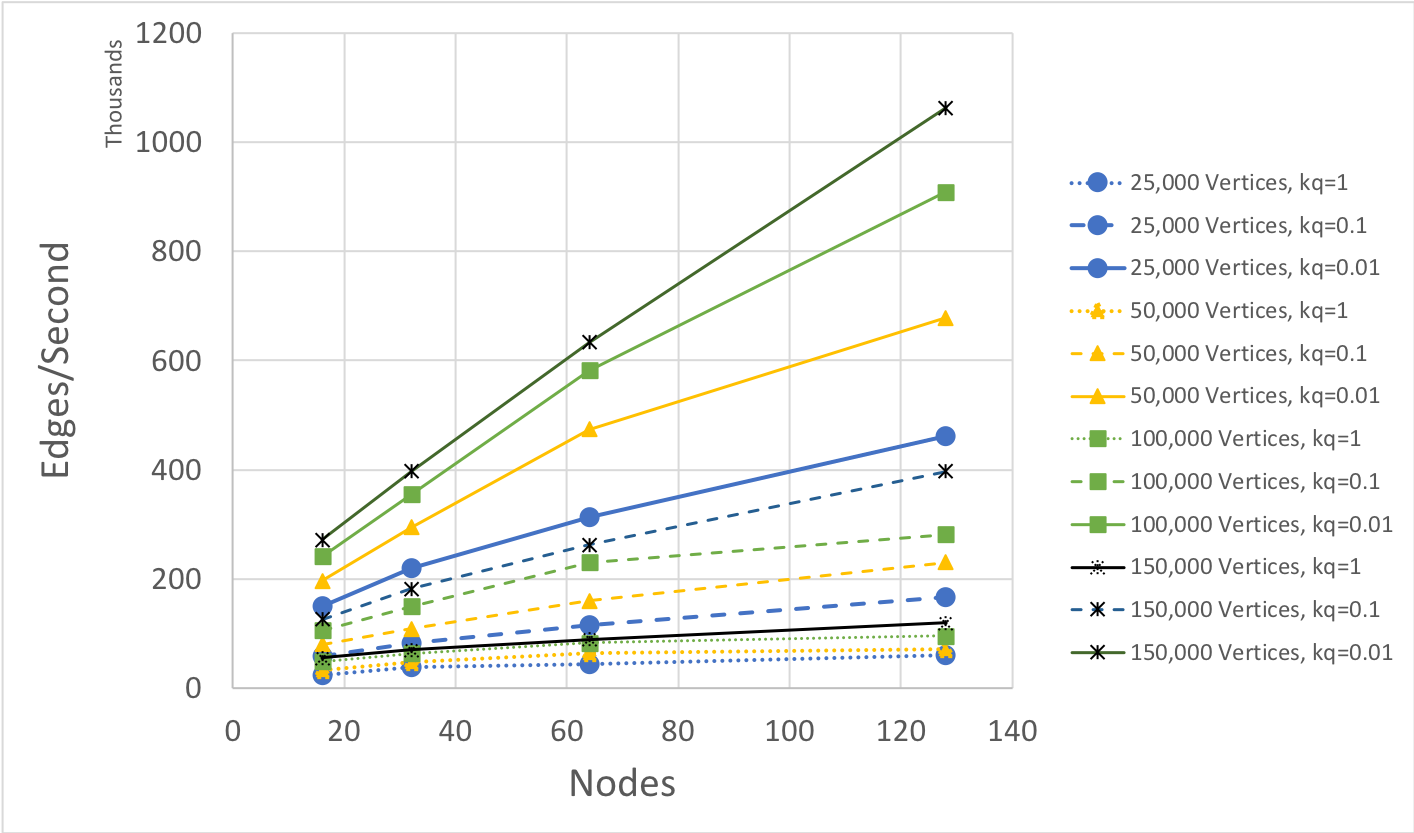}
\caption{We keep the first edges of queries ($kq$) with probabilities
$\{0.01, 0.1, 1\}$.  We also vary $|V|$ to range from 25,000 to
150,000.  This figure is a plot of edges/second versus the 
number of nodes.}
\label{figure:keep_queries}
\end{figure}


\section{Related Work}
\label{section:related}

The specification of the subgraph in SAL is influenced by SPARQL
\cite{standard:w3c:sparql}, namely the topological structure specification.
Semantic web research has produced several approaches to streaming
RDF data and continuous SPARQL evaluation, namely
C-SPARQL \cite{c-sparql-2010}, Wukong \cite{wukong-2016}, 
Wukong+S \cite{wukong-2017}, EP-SPARQL \cite{ep-sparql-2011}, 
$SPARQL_{Stream}$ \cite{sparql_stream-2010},
and CQELS \cite{cqels-2011}.  While the streaming aspect of
this work is similar to our own goals, there are several aspects not
addressed which our work targets, namely the temporal ordering
of the edges, computing attributes on streaming data, and edges
with multiple attributes. 
Finally, with the exception of Wukong+S \cite{wukong-2017} that report
scaling numbers to 8 nodes, none of the streaming RDF/SPARQL 
approaches are distributed.

Recently, many frameworks have been
developed that provide streaming APIs, including 
Storm, 
Spark, 
Flink, 
Heron, 
Samza, 
Beam, 
Gearpump, 
Kafka, 
Apex, 
Google Cloud Dataflow.
Many could be the basis of a backend implementation
for SAL.  
In Section \ref{section:results} we compared Flink with SAM
and found that each framework excelled in different
parameter regimes.  We leave as future work an exploration
into combining the strengths of each approach in a unified
setting.

Graph theory is an area rich with results, including a number of DSLs 
such as Green-Marl
\cite{green-marl-2012}, Ligra \cite{ligra-2013},
Gemini \cite{gemini-2016}, Grazelle \cite{grazelle-2018}, and
GraphIt \cite{graphit-2018}.  Except for Gemini, the implementations
of these DSLs target shared memory architectures.  
Regardless of whether these DLS approaches can be adjusted
to a distributed setting, a fundamental difference between this set
of work and our own is the streaming aspect of our approach and domain.
An underlying assumption of these DLSs is that
the data is static.  The algorithmic solutions, partitioning,
operation scheduling, and underlying
data processing rely upon this assumption.

A plethora of efforts using vertex-centric computing
exist to perform bulk synchronous parallel graph computations.
Some run only on one node,
\cite{venus-2015,
flash-graph-2015, g-store-2016}
while some support distributed computations, 
\cite{ gps-green-marl-2014,
powerlyra-2015, graphd-2018, pregelix-2014, chaos-2015}.
Some support only in-memory computations
\cite{gps-2013,powerlyra-2015} while others can utilize out-of-core
\cite{venus-2015,
 chaos-2015, flash-graph-2015, g-store-2016}.
None of the above methods provide any native
support for temporal information.  However,
it would be possible to encode temporal information
on edge data structures. 
The biggest impediment to using vertex-centric computing for a streaming
is that it is by nature a batch process.  Adapting it to a streaming
environment would likely entail creating overlapping windows of time
over which to perform queries in batch.  

\section{Conclusions}
\label{section:conclusions}

In this paper we presented a domain specific language for expressing
temporal subgraph queries on streaming data named the Streaming
Analytics Language or SAL.  We also presented the performance
of an implementation of SAL called the Streaming Analytics Machine
or SAM.  We explored finding temporal triangles, and compared SAM to a custom solution
written for Apache Flink.  SAM excels when triangles occur frequently
(greater than 5 per second per node) while Flink is best when
triangles are rare.  For all parameter settings, SAM showed improved
throughput to 128 nodes or 2560 cores, while Flink's rate started
to decrease after 32 nodes.
For more selective queries, SAM is able to obtain an aggregate rate
1 million netflows per second.  While we have presented SAM as the
implementation for SAL, Apache Flink shows promise as another potential
underlying backend.  Future work understanding the performance differences
between both platforms will likely be beneficial for both approaches.
We also have further evidence that programming in SAL is an efficient
way to express streaming subgraph queries.  SAL requires only 13
lines of code, while programming directly with SAM is 238 lines
and the Flink approach is 228 lines.

\bibliographystyle{abbrv}
\bibliography{bib/bib,bib/streaming-algorithms,bib/sparql,bib/data-stream-management-systems,bib/graph_vertex-centric,bib/graph_dsls}

\end{document}